\documentclass[12pt,ngerman]{article}
\usepackage{amsmath,amsthm,amssymb,amscd}

\newtheorem{cor}{Corollary}[section]

\newtheorem{theorem}[cor]{Theorem}

\newcommand{\e}{{\rm e}}
\newcommand{\R}{\mathbb{R}}

\newcommand{\B}{\mathcal B}
\newcommand{\M}{\mathcal M}
\newcommand{\Prm}{\mathcal P}

\newcommand{\refeq}[1]{(\ref{#1})}

\usepackage{pgf}
\usepackage{tikz}

\begin{document}
\title{Neuronal calculus for the auditory pathway}
\author{Daniel Aalto$^1$, Hans Martin Reimann$^2$ and Eero Saksman$^3$}
\maketitle

\vspace{12cm}
{$^1$ Department of Communication Sciences and Disorders, University of Alberta, 2-18 Corbett Hall, Edmonton, AB T6G 2G4, Canada}\\
{$^2$ Mathematisches Institut, Universit\"at Bern, Sidlerstrasse 5, CH-3012 Bern, Switzerland}\\
{$^3$ Department of Mathematics and Statistics, University of Helsinki, PO Box 68, FI-00014, Finland}

 \vspace{3cm}
\newpage

\section{Abstract}
The first steps in the neural processing of sound are located in the
auditory nerve and in the cochlear nuclei.
To model the signal processing efficiently, we propose a simple mathematical
tool that takes the minute timing of the system into account.
In contrast to the situation in the cortex,
the number of connections between neurons in auditory periphery
is comparatively low. This gives way to an accurate modeling of the connectivity of the neuronal network. The timing is the all important feature in the peripheral neuronal auditory pathway.
The primary auditory neurons e.g. phase lock to periodic sounds
with important interactions with respect to both the refractory periods of the neurons and to
the time delays caused by traveling times along the basilar membrane or through a synaptic
connection.

The mathematical tools provide a solid basis to build models for peripheral auditory processes.
In particular,
we study carefully a large class of refractory neurons, find analytical formulas for the spiking activity,
and prove that refractory neurons respond to periodic signals by asymptotically periodic output.
The methods rely on
the theory of positive operators and give a numerical scheme for finding fixed points to an integral operator
with geometric convergence rate.
In addition,
we consider a perfect integrator neuron, mathematically equivalent to randomized random walk,
where the random walk is bounded from below, and solve the first passage time problem
using continuous time Markov chain techniques. Our method leads to ordinary differential equations that are linear. The dynamical
behavior can thus be described by classical methods.

In an accompanying paper we set up the simulation framework as a counterpart to the present mathematical model. By suitably adjusting the few parameters in the model it is possible to reproduce the basic patterns of neural activity.

\section{Introduction}

In this paper we try to develop a mathematical tool
that allows to describe signal analysis in the auditory pathway of the brain.
The incoming signals from the inner ear are processed in parallel pathways up to the inferior colliculus.
Essentially starting from this neuronal center information from other parts of the brain are combined
with the auditory components.
A time dependent pattern of neural activity then emerges in the auditory cortex.
Signal processing along the auditory pathway is complex.
There is a wealth of detailed experimental information available
and many of the essential features have been described precisely \cite{inferiorcolliculus}.
To a large extent the activity of specific neurons has been classified
and some information on the topological ordering of different groups of neurons is available.
The challenge is to set up a global picture
and to understand
how the different components combine on the signal processing level.\\
The approach taken here is
to describe the neuronal activity with densities that have a probabilistic interpretation.
The neuronal nuclei are then pictured as transforming these densities in a specific way,
adapted to the physiological significance of these nuclei.
Densities give a less precise description of neurons than models based on dynamical systems would admit.
Yet densities should provide a more efficient tool in the treatment of the interaction of the neuronal nuclei
and they should facilitate keeping track of the patterns of neural activity.
They are specifically designed towards time and intensity coding.

We land at the area of firing rate models (see chapter 11 of \cite{ermentrout2011}) and choose a Poisson type model that allows to represent the firing properties of neurons on the basis of densities. We first analyze the signal processing of a single neuron with respect to the refractory properties. Such effects have already be studied previously
 in \cite{deger}.

 Our method is based on stochastic results from renewal theory. With it, any kind of refractory period can be handled.

For periodic signals we show that any
refractory neuron responds to a periodic input by an asymptotically periodic output.

To model the combined activity of a network of refractory neurons, we use perfect integrate-and-fire with
stochastic input and bounded paths. It turns out that a continuous time finite state Markov chain model (CTMC)
is suitable to the analysis. In particular, the synaptic and transmission delays are easily incorporated.

This approach circumvents Stein's model and stochastic differential equations.  Yet still, it seems to capture the essential features observed with the neurons in the auditory pathway. Our model leads to an
ordinary differential equation that is linear. Explicit solutions can then
be obtained through standard techniques in linear algebra and differential equations.
We e.g. get an easy derivation for the classical formulas
for perfect integrate-and-fire neurons receiving only excitatory inputs.

Our mathematical approach is amenable to modeling and relates to simulations.
In a companion paper this will be investigated in detail.

In order to focus our work, we have asked for transparency and simplicity of the model and our efforts are directed towards using as few parameters as possible.
We take into account the fine timing parameters like the refractory properties of the neuron and the synaptic and
transmission delays. We also include the spontaneous activity of the neuron and
the local architecture of the network.
We omit any neurochemical variables and hence the model is at best only phenomenological.
\\

\section{Stochastic processes for auditory periphery}

A neuron receives spike trains through its dendritic tree and emits a spike
whenever the membrane potential of the cell reaches a cell specific threshold value.
The spikes are narrowly supported in time and have uniform shapes.
Thus the mathematical theory of point processes is widely used for modeling the neurons statistical behavior.
In particular the problems related to the neuronal ensembles can be reformulated in terms of queueing theory.

There are models for single neuron behavior of varying detail.
The most detailed ones are based on the physical model by Hodgkin and Huxley
and contain several parameters related to the chemical and physiological properties of the neuron.
From the mathematical point of view, this model describes the voltage of the membrane in terms of a dynamical system,
i.e. a system of ordinary differential equations.
The simplest model is the integrate-and-fire (IF)
and a more sophisticated is the leaky IF (LIF), which breaks the problem of neuronal firing in two parts:
the time before the neuron fires is modeled with differential equations and once the solution reaches critical level,
the neuron fires and then rests before the process starts anew.

The classical models are deterministic. Stochastic nature of the neural activity can be incorporated in several different ways.
The Stein model \cite{stein65} is a stochastic differential equation model for the sub-threshold membrane voltage evolution.
It assumes excitatory and inhibitory Poisson inputs and an exponential decay of the membrane potential.
This is still a very detailed model and correspondingly, explicit calculations are difficult, e.g. an explicit formula for the first passage time of the system with constant
intensity Poisson inputs is not known (see Sacerdote, Giraudo 2011).
If the leakage term is neglected, the resulting model is randomized random walk (RRW)
for which the first passage time can be derived - assuming constant input - via Laplace transforms \cite{tuckwell89}.
However, there is no formulas available for non-constant stimuli.\\

In the present approach we take every incoming signal to be inherently stochastic by assuming always
an explicit stochastic process,
which comes with a certain intensity
depending on the (acoustic) stimulus and previous processing steps only. The mathematical expression for this intensity is a non-negative time dependent density $s(t)$.
We assume, that the neurons activity is completely described by the nature of the
stochastic process (e.g. Poisson, Gamma type neuron), by the density $s(t)$ that itself is derived from the input signals from other neurons and by
the time of the previous spike of the neuron itself.
To easily combine the activities of the different neurons, we only compute the output intensity
of the spiking and assume that when several of these intensities are combined in a next neuron, the
pooled incoming spikes look like a Poisson process. In the case of constant intensity
processes which are sufficiently regular, the sum of the processes
approaches Poisson process (see \cite{daley}: Proposition 11.2.VI).

We are interested in two types of abstract neurons:
primary neurons are directly stimulated by a continuous variable
like the concentration of neurotransmitter in the
inner hair cell auditory nerve fiber complex -
or a more abstract variable like a probability density;
integrating neurons take electric spike trains from other neurons as input
and process the spike trains according to a rule where the
excitatory input brings the neuron closer to firing and inhibitory
spikes push the neuron away from emitting a spike. The rules are motivated by
the physiological properties of the cells.
We try to characterize the quantitative behavior of these abstract
neurons based on first principles.

\section{The refractory neuron}

\subsection{The definition of the refractory neuron}

An abstract refractory neuron, $n_R$, is taken to be
a stochastic process determined by an inhomogeneous, non-negative function
$s(t)$ and a homogeneous function $r(t)$.
The inhomogeneous function represents the external stimulation of the neuron
and the homogeneous function describes the internal dynamics of the neuron.
As an example, $s(t)$ could describe the amount of the neurotransmitter in the inner hair cell,
and $r(t)$ the refractory properties of the auditory nerve fiber.
More precisely we define the probability of the neuron to emit
a spike infinitesimally through
\begin{equation}
\label{n_R: infinitesimal definition}
P(t_1>t+h|t_1>t)=1-s(t)r(t-t_0)h + O(h^2)
\end{equation}
where $t_0$ is the moment at which the neuron emitted the previous spike before $t_1$
and $O(h^2)$ is an error term with $O(h)/h\to 0$ as $h\to 0$.
We assume throughout
that $s(t)$ is a non-negative locally integrable function, $r(t)=0$ for all negative $t$,
and $r(t)\leq r(u)$ for all $t<u$.
Furthermore we assume that the integral
$$
\int_x^{\infty}s(t)r(t-x)dt
$$
is infinite for any $x>0$.

To determine the probability density of the firing we fix $t>t_0$ and divide the interval $[t_0,t]$ into $n$ equal parts
with $nh=t-t_0$, $a_j = t_0+(j-1)h$ and $b_j = a_j+h$, $j=1,\dots, n$, to have
$$
[t_0,t]=\bigcup_{j=1}^n[a_j,b_j].
$$
By the definition of the conditional probability we have
\begin{eqnarray}
\nonumber
P(t_1>t|t_1>t_0)&=&\frac{P(t_1>t)}{P(t_1>t_0)}\\
\nonumber
&=&
\frac{P(t_1>b_n)}{P(t_1>a_n)}\cdot \ldots \cdot \frac{P(t_1>b_1)}{P(t_1>a_1)}\\
\nonumber
&=&
\prod_{j=1}^nP(t_1>b_j|t_1>a_j).
\end{eqnarray}
Taking logarithm, using the assumption (\ref{n_R: infinitesimal definition}) and linearizing the right hand side logarithms gives
$$
\log P(t_1>t|t_1>t_0)=\lim_{n\to \infty}\sum_{j=1}^n -hs(a_j)r(a_j-t_0)
$$
which can be interpreted as an integral to give
$$
P(t_1>t|t_1>t_0)=\exp \left(-\int_{t_0}^ts(u)r(u-t_0)du\right).
$$
Let us define now the transition probability $p(x,t)$ by differentiating
the conditional probability above with respect to $t$ and setting $t_0=x$ to obtain
\begin{equation}
\label{pxt}
p(x,t)=s(t)r(t-x)\exp \left(-\int_{x}^ts(u)r(u-x)du\right).
\end{equation}
Then $p(x,t)$ is the probability density that the first firing after firing at $t_0=x$ occurs at $t$. We set
\begin{equation}
p(x,t) = 0\quad\textrm{for}\quad t<x.
\end{equation}
Clearly\footnote{We often use the shorthand $\int$ for the definite integral $\int_{-\infty}^{\infty}$, and likewise we abbreviate $\sum=\sum_{-\infty}^{\infty}$, if not otherwise stated or clear from the context.}
\begin{equation}
\int p(x,t)dt = 1
\end{equation}
If the ``first'' firing time has a distribution $p_0$, then the next
firings are given recursively by
$$
p_k(t)=\int_{-\infty}^tp_{k-1}(x)p(x,t)dx,
$$
for $k=1,2,\dots$.
The consecutive firing times form a point process $T_R=\{t_0,t_1,\dots\}$.\\
Analogously, the transition probabilities for the $k$th firing are
$$
p_k(x,t)=\int_{-\infty}^tp_{k-1}(x,z)p(z,t)dz,
$$

Similar derivations can be found in \cite{gerstnerkistler}: 5.2.3,
and \cite{cox}: pp. 4--5.

\subsection{The output of the refractory neuron}

Given either a refractory neuron or a integrate-and-fire neuron as a model for the ``law''
of the neuron, we define the output rate of the neuron.
Let  $T = \{t_0,t_1,\dots\}$ be a point process.
The attached counting process is defined by
$$
N(t) = \{i\geq 0|t_i\leq t < t_{i+1}\},
$$
and we call $M(t) = E(N(t))$ the expected number of neural spikes emitted up to moment $t$.
If $M$ is differentiable, we define
$$
I(t) = \frac{dM(t)}{dt}
$$
to be the instantaneous firing rate. This is the time varying
output rate of the neuron's activity.
In the physiological measurements, the peri stimulus time histogram
(PSTH) corresponds to the instantaneous firing rate. On the other hand,
in stochastic analysis $M(t)$ is called the renewal function.

Mathematically, in the  simplified model treated in this paper, a simple neuron is the transformation of incoming
rate functions $s_-$ and $s_+$ to an output rate function $I(t)$.

\subsection{Instantaneous firing rate}

Given a locally integrable non-negative function $s(t)$ and an initial probability
distribution $\mu(dt)=p_0(t)dt$ for the first spike, we compute the
instantaneous firing rate $I(t)$.

The instantaneous firing rate $I(t)$ of a neuron at time $t$ corresponds
to the joint probability density function $p(x,t)$ of the neuron to fire
at moment $t$ given it fired at $x$ and
the initial probability measure $\mu$ due to the history of the system
up to a starting moment $t=0$ (stimulus onset).
We have (justified by \cite{gihman69}: Chapter 7.7: Theorem 4)

\begin{eqnarray}
\label{instantenous}
I(t) &=& \sum_{k=1}^{\infty}\int_{t_1=0}^{t_2}\dots \int_{t_k=t_{k-1}}^t\int_0^{\infty}
\mu(s)p(s,t_1)\dots p(t_k,t)
dt_1\dots dt_kds\\
&=& \sum_{k=1}^{\infty}p_k(t)
\end{eqnarray}
The definition of instantaneous firing rate is related to the repeated measurements data gathering
procedure,
where the activity of single neurons
is recorded while the auditory system receives acoustic input.
It is tacitly assumed that between the stimuli there is enough time for the
neural system to get back to a equilibrium state
where only the spontaneous activity of neurons persists.
If the spontaneous activity is assumed to be a homogeneous Poisson process, then
the corresponding measure $\mu$ would be just the exponential function
$$
\mu (s) = \lambda \e ^{-\lambda s}ds.
$$
In a mathematically simpler situation, we can assume that the neuron fired at $t=0$, at the onset
of the stimulus. This corresponds to the choice
$$
\mu (s) = \delta (s),
$$
where $\delta$ is the Dirac measure concentrated at the origin.

The case with no refractory period ( $r(t)= \chi_0 (t)$ )
corresponds to the classical Poisson model. In this situation, the instantaneous firing rate
gives back the original intensity

\begin{eqnarray*}
p_m(x;t)
&=&dt \chi_x(t)\,s(t) \,\exp(-\int_x^t s(u)du)\\
& &\int_x^t s(y_{m-1}) dy_{m-1}...\int_x^{y_3} s(y_2) dy_2\int_x^{y_2} s(y_1) dy_1
\end{eqnarray*}
Upon setting $w(t)= \int_x^t s(u) du$,
$s(t)= \frac{d}{dt}w(t)$, one obtains
\begin{eqnarray*}
\int_x^t s(y_{m-1}) dy_{m-1}...\int_x^{y_3} s(y_2) dy_2\int_x^{y_2} s(y_1) dy_1 &=&\frac{w^{m-1}(t)}{(m-1)!}
\end{eqnarray*}
and the transition probabilities have the form
\begin{equation}
p_m(x;t)
= \chi_x(t) \,s(t)\e^{-w(t)}\frac{w^{m-1}(t)}{(m-1)!}.
\end{equation}

The instantaneous firing rate of a neuron at time $t$, given that the neuron fired at time x, is
\begin{equation}
q(x;t)=\sum_{m=1}^{\infty}p_m(x;t)=\chi_x(t)s(t)\e^{-w(t)}\e^{w(t)}
=\chi_x(t)s(t).
\end{equation}
We thus
obtain the original intensity $s$ and observe that the instantaneous firing
rate does not depend on the past. Hence,
$$
\lim_{x\to -\infty}q(x;t)=s(t)
$$
for all $t$
and the convergence is uniform on any finite interval.
In the following sections neurons with non-trivial refractory period will be considered.

\subsection{Densities}

At the level of single neurons, we assume that the activity of a neuron is determined
by the incoming densities $s_i(t) \geq 0$
and a set of parameters: spontaneous activity of the neuron $\sigma$, connection strength $\omega _i$,
connection delay $\tau_i$,
refractory function $r$ and firing threshold $\vartheta$.
We first divide the analysis of the neuron into two independent steps:
input analysis and output generation. The first results in a single pooled activity,
a Poisson process that we approximate by the integrated input density $s$.
The output generation depends on  $s$ and $r$, the refractory component, as well as on $\sigma$ and $\vartheta$. The spontaneous activity $\sigma$ can be subsumed in the calculation of $s(t)$.\\

The infinitesimal probability density that determines the instantaneous firing rate of the output is a product of the integrated input density $s(t)$, and the refractory function $r(t-x)$, translated to the position $x$ of the previous firing time
(see \cite{lutkenhoner80}, \cite{gaumond82}, and \cite{miller92}, the details are given in the next section).

A typical example for $r$ is given by
\begin{equation}
\label{rho_exp}
r_{\rho}(u) = \left\{
\begin{array}{rl}
0 & \text{if } u<\rho_A,\\
1-\exp (-u/{\rho_R}) &\text{if} u \geq \rho_A,
\end{array} \right.
\end{equation}
where the constants $\rho_A$ and $\rho_R$ are absolute and relative refractory periods respectively.
We assume that $r$ is always monotone increasing, grows at most polynomially,
and $r(u)=0$ for all $u<0$. We do not assume continuity for $r$.\\
For the incoming densities we consider two cases:
either we have an auditory nerve fiber, for which a simple inner hair cell model will then provide the input density,
or the input comes from a collection of neurons with specific outputs. In the latter case, we
represent the activities of the neurons that connect over the dendritic tree by
densities $s_i, i\in I$.\\
The (excitatory) inputs are pooled
by taking the sum of the densities $s_i$, weighted with the strength $\omega_i$ of the connections,
and taking into account the delays $\tau _i$ of each path,
\begin{equation}
s(t)=\sum_{i\in I} \omega_i s_i (t-\tau_i),
\end{equation}

where $I$ is the set of indices for the (excitatory) connections.\\

For the refractory neuron only excitatory inputs are taken into account. The situation of mixed, excitatory and inhibitory inputs, will be studied in the context of the integrate-and-fire neuron.\\

The pooling of the densities is motivated by the following addition property\\
\emph{
Assume that the neuron $Y$ fires whenever any of the two input neurons $X_1$ or $X_2$ fire. If the firing densities for $X_1$ and $X_2$ are $s_1$ and $s_2$ respectively, then the resulting output density of the neuron $Y$ is $\tilde{s}= s_1 + s_2$}\\
This statement fails, if refractory periods are involved.\\
The proof is a direct consequence of the exponential law:
\begin{eqnarray}
P(T_1^Y>t|T_0<x)
&=&
1-P(T_1^{X_1}\leq t| T_0^{X_1}<x)P(T_1^{X_2}\leq t|T_0^{X_2}<x)
\nonumber\\
&=&
1
-
\exp \left(-\int_x^ts_1(u)du\right)
\exp \left(-\int_x^ts_2(u)du\right)
\nonumber\\
&=&
1
-
\exp \left(-\int_x^t(s_1(u)+s_2(u))du\right).
\nonumber
\end{eqnarray}
Hence, $Y$ can be obtained by summing the densities of $X_1$ and $X_2$.\\

\subsection{Periodic signals}

The transition probability (see equation (\ref{pxt})) associated to the neuron and to the density $s(t)$ is
$$
p(x,t) = r(t-x)s(t)\e ^{-\int_x^tr(u-x)s(u)du}.
$$
For the periodic case it is assumed that $s$ is 1-periodic. As a consequence
$$
p(x+m,t+m) = p(x,t)
$$
for any integer $m$.\\

In the periodic case, the instantaneous firing rate is well defined and asymptotically does not depend on the initial measure $\mu$. This is a result of Thorisson (\cite{thorisson84}, Theorem 2). In the following we use $\mu = \delta _x$ the dirac measure at $x$. With this,
the instantaneous firing rate of a neuron at $t$ is given by
\begin{equation}
q(x,t)=\sum_{m=1}^{\infty}p_m(x,t)
\end{equation}

Under the additional condition that the factor m defined in equation (\ref{m})  below is finite, the limit

$$
q(t):= \lim_{x\rightarrow -\infty} q(x,t)
$$
exists (\cite{thorisson84}, Theorem 6).\\
In the discussion above, when no refractory period was involved, it was shown that $q(t)$ was equal to the initial density $s(t)$. In the present situation this will be different.\\
Note that the function $q$ is periodic:
$$
q(x,t+1) = \sum_{m=1}^{\infty}p_m(x,t+1)=\sum_{m=1}^{\infty}p_m(x-1,t) = q(x-1,t)
$$
$$
q(t+1)= \lim_{x\rightarrow -\infty} q(x-1,t)= q(t)
$$

The operator $T$ on locally integrable periodic functions is defined as
$$
Tf(t) = \int p(x,t) f(x) dx
$$
and its dual on $L^{\infty}$ by
$$
T^*g(x) = \int p(x,t) g(t) dt
$$
The fact is then that $q$ is invariant under $T$
\begin{eqnarray*}
Tq(t)&=&\int p(x,t) q(x) dx = \int p(x,t) \lim_{z\rightarrow -\infty} q(z,x) dx\\
&=&\lim_{z\rightarrow -\infty} \int p(x,t) \sum_{m=1}^{\infty}p_m(z,x)\\
&=&\lim_{z\rightarrow -\infty}  \sum_{m=1}^{\infty}p_{m+1}(z,t)\\
&=&\lim_{z\rightarrow -\infty}  (q(z,t) -p(z,t))\\
&=& q(t)
\end{eqnarray*}
\\

The operator T can be periodized.
\begin{eqnarray*}
Tf(t) &=& \int p(x,t) f(x) dx\\
&=& \sum_m \int^1_0 p(x+m,t) f(x+m) dx\\
&=&  \int^1_0 \sum_m p(x+m,t) f(x) dx
\end{eqnarray*}
Set
$$
\tilde{p}(x,t) = \sum_m p(x+m,t) = \sum_m p(x,t+m)
$$
then the operator $T$ on periodic functions reduces to
$$
\tilde{T} f(t) := \int^1_0 \tilde{p}(x,t) f(x) dx
$$
on the space $L^1 [0,1]$. Its kernel is periodic in both variables and satisfies
$$
\int^1_0 \tilde{p}(x,t) dt = \int^1_0 \sum_m p(x+m,t) dt = \int^1_0 \sum_m p(x,t+m) dt\\
= \int p(x,t) dt = 1
$$
Furthermore, $q$ considered as a function in $L^1 [0,1]$ is invariant
$$
q(t) = T q(t) = \tilde{T} q(t)
$$

In the case of a $1$-periodic infinitesimal density function it is possible to start with the probability distribution $\tilde{Q}(A)$ that the neuron has fired at $t_0 \in A$ mod 1. The set $A$ is a subset of $[0,1)$ and $\tilde{Q}([0,1)) = 1$. Assume then that the neuron has fired at time $t_0=x$ within the period.

The initial probability $\tilde{Q}$ is transformed into the probability

$$ \int_A dt \int
\tilde{p}(x,t) dQ(x)$$

If $\tilde{Q}$ is mapped onto itself under this transformation, then it can be represented by a density $\tilde{q} dt$
$$\tilde{Q}(A) = \int_A \tilde{q}(t) dt$$
that satisfies the fixed point equation
$$
\tilde{q}(t) = \int_0^1\tilde{p}(x,t)\tilde{q}(x)dx.
$$
The output density $q(t)$ of the neuron driven by the periodic input density $s(t)$ is a fixed point of $T$ and hence of the periodized operator $\tilde{T}$. In general it will not be normalized whereas the fixed point $\tilde{q}$ above is normalized by $\int_0^1 \tilde{q}(t)dt = 1$.

\begin{theorem} (Thorisson)\label{fixedpoint}
The integral operator
$$
\tilde{T}q(t) = \int_0^1\tilde{p}(x,t)q(x)dx.
$$
has a unique normalized fixed point  $\tilde{q} \in L^1([0,1])$. Furthermore
for every non-negative $q_0 \in L^1([0,1])$ with $\|q_0\|=1$ the sequence
$q_n = T q_{n-1}, n=1,2,...$  converges exponentially to $\tilde{q}$ in $L^1([0,1])$.
\end{theorem}

This result is contained in \cite{thorisson84}. An independent proof will be given
in Appendix 2 (Corollary \ref{adjoint}).

The uniqueness statement tells us in particular that the
outgoing infinitesimal probability distribution $q$ of the neuron - provided it exists - is a multiple of the normalized distribution $\tilde{q}$.\\
The expectation time for $t_1$ to occur on the condition that $t_0=x$ is
$$\int t p(x,t) dt$$
The expected delay is thus
$$E(x):=\int t p(x,t) dt-x =\int (t -x) p(x,t) dt$$\\

\begin{theorem} (Thorisson)
\label{constant}
Assume that
\begin{equation}
\label{m}
m=\int_0^1 E(x) \tilde{q}(x)dx <\infty
\end{equation}
where $\tilde{q}$ is the normalized fixed point above. Then the outgoing infinitesimal probability distribution approaches the periodic density
$$
q=\frac{1}{m} \tilde{q}.
$$
This entails that $\displaystyle\lim_{n\to\infty} \int_n^{n+1}|q(x,t)-\tilde q(t)|dt=0$ for every $x$.
\end{theorem}
This is proved in \cite{thorisson84}, Theorem 3.
As an example take the case of a neuron with an absolute refractory period $\varrho$ and a constant input density $s(t)= A$. Here,
$q(t)=1$ is a fixed point for $\tilde{T}$. The expected average delay is
\begin{eqnarray*}
\int_0^1 E(x) dx &=& \int_0^1 dx \int_{x+\rho}^\infty(t-x) Ae^{-A(t-(x+\rho))} dt\\
&=&\rho + A^{-1}
\end{eqnarray*}
The outgoing density is thus
$$q(t) = \frac{1}{\rho+A^{-1}}$$
in accordance with the result in Appendix 1.\\

\subsection{Constant stimuli with general refractory structure}

We consider an important special case which corresponds to a neuron with
spontaneous firing and a short refractory time. Typically the spontaneous
firing rate is at most 100~Hz while the refractory time is at least 0.7~ms.
We model the situation by taking a constant density.

In the special case $s(t)\equiv A$, $A>0$, the transition probability
\begin{equation*}
Ar(t-x) e^{-A\int_x^tr(u-x)du}
\end{equation*}
is a function of $t-x$ that will be written as $p(t-x)$.

The transition probabilities for the k-th firing are then obtained by convolution
\begin{eqnarray*}
p_k (t-x)&=& \int p_{k-1}(t-z)p(z-x)dz\\
&=&p_{k-1}*p(t-x) = p^{*k}(t-x)
\end{eqnarray*}
Provided the absolute refractory period is positive,
the instantaneous firing rate
$$
q(t-x)=\sum_{m=1}^{\infty}p_m(t-x)
$$
with system start at $x$ has only finitely many terms.
The instantaneous firing rate of the neuron is then given by
\begin{eqnarray*}
q&=& \lim_{x\rightarrow-\infty} \sum_{1} ^{\infty}p_k(t-x)\\
&=& \lim_{t\rightarrow \infty} \sum_{1} ^{\infty}p^{*k}(t)
\end{eqnarray*}
In this situation one actually need not invoke Thorisson's theory, as the  classical Blackwell renewal Theorem \cite{F2} verifies that the previous limit exists and is constant.

In the case of a pure absolute refractory period $\rho > 0$ we already noted that
\begin{equation*}
q=\frac{A}{1+A\rho}.
\end{equation*}
Another proof that uses just simple Fourier analysis in included in Appendix 1.

\section{The integrate--and--fire neuron: a simple discrete model}

In this section we derive the instantaneous firing rate and the interspike interval
distributions for the integrate--and--fire neuron. We analyze the example of constant
inputs and finally show that for every periodic stimulus there exists a unique
periodic equilibrium density which the neuron's instantaneous firing rate approaches.
This simple model does not allow the inclusion of a refractory period.

\subsection{Integrate-and-fire neuron}

An integrate-and-fire neuron (IF neuron),
$n_{IF}$ is defined through inhibitory and excitatory incoming spike trains.
The model described here is perfect in the sense that the effect of the incoming spikes does not
decay over time. Moreover, the model has bounded paths since the neuron is not allowed to
have infinitely big membrane potential values. This is done by modeling the membrane potential
in an abstract way as states in a finite system, where each state characterizes how many excitatory
incoming spikes are needed at least before the neuron can emit a spike.

More precisely,
let $s_-(t)$ and $s_+(t)$
be the intensities of two independent, inhomogeneous Poisson processes.
The IF neuron has $K>1$ possible states
$1,\dots,{K}$. When a spike arrives from the inhibitory process,
the IF neuron moves from state $i$ to state $i+1$
unless already at the lowest state $K$.
Similarly, when a spike arrives from the excitatory process, the neuron moves
up from the state $i$ to $i-1$, if $i>1$ and
moves from state $1$ to $K$ otherwise. This transition is called resetting
and the IF neuron emits a spike during the transition.
These transition times form a point process $T_{IF}=\{t_0,t_1,\dots\}$. 

\subsection{Densities}\label{subse:densities}
We assume that the activity of the neuron is determined
by the incoming densities $s_i(t) \geq 0$
and a set of parameters: spontaneous activity of the neuron $\sigma$, connection strength $\omega _i$,connection delay $\tau_i$ and firing threshold $\vartheta$.
We divide the analysis of the neuron into two independent steps:
input analysis and output generation. The first results in a single pooled activity,
a Poisson process that we approximate by the input density $s$. The output generation depends on $s$ as well as on $\sigma$ and $\vartheta$. The spontaneous activity $\sigma$ can be subsumed in the calculation of $s(t)$.\\
The incoming densities $s_i, i\in I$.
represent the activities of the neurons that connect over the dendritic tree.
The excitatory and inhibitory inputs are pooled separately
by taking the sum of the densities $s_i$, weighted with the strength $\omega_i$ of the connections,
and taking into account the delays $\tau _i$ of each path,
\begin{equation}
s^+(t)=\sum_{i\in I_+} \omega_i s_i (t-\tau_i),
\end{equation}
\begin{equation}
s^-(t)=\sum_{i\in I_-} \omega_i s_i (t-\tau_i),
\end{equation}
where $I_+ \subset I$ and $I_-\subset I$ are the subsets of indices for excitatory and inhibitory connections respectively. The probability of an incoming spike is modeled by the density $s(t)= s^+(t) + s^-(t)$.\\

\subsection{The infinitesimal generator}
A probability vector $v = (v_1,v_2,..., v_n)$ is a vector with non-negative components that add up to $1$: $v_i\geq 0$ for $i=1,...n$ and $\sum_1^n v_i = 0$. A linear mapping that maps probability vectors into probability vectors is called a probability mapping. Expressed by a matrix $A = (a_{ij})$ in standard coordinates it is characterized by the condition that its column vectors are all probability vectors. If the linear differential equation
$$v'=Qv$$
generates a flow of probability mappings, then $Q$ is called an infinitesimal generator for probability mappings. It is well-known that the infinitesimal generators $Q$ can be characterized by the property that they can be written in the form $Q = (A - I)s$, with $A$ a probability matrix and $s\geq 0$  is a scalar. In other words, the diagonal-elements of $Q$ are non-positive, other elements non-negative, and each column sums to zero. Note that the flow of the differential equation maps the positive cone $\mathbf{R^n}_+ = \{v: v_i\geq 0, i=1,...n\}$ into itself. In fact, on the boundary of the cone, the vector field $Qv$ points into the cone or is tangent to the cone.\\

The infinitesimal generator of the process described in subsection \ref{subse:densities} is
given by the matrix
$Q(t) = (A(t) - I) s(t)$ with a matrix $A(t)$ of the form
\begin{equation*}
A(t) =
\begin{bmatrix}
0 & p(t)  &   & & &  \\
q(t)  & 0 & p(t) & & &  \\
  &      \ddots  & \ddots  & \ddots& \\
  &      &       q(t)&    0& p(t) \\
p(t) &      &            &  q(t) & q(t) \\
\end{bmatrix},
\end{equation*}
with $p(t)= \frac{s_+(t)}{s(t)}$ and $q(t)=\frac{s_-(t)}{s(t)}$.

The time development of the system is fully determined by the
matrix differential equation
$$
v'(t) = Q(t)v(t),
$$
where $v$ is a time-dependent probability vector.

A first order linear matrix equation has always a unique solution with prescribed initial condition $v(x)= w$. It can be expressed by the Peano-Baker series
in terms of iterated integrals, see \cite{rindos}, \cite{fortmann}, \cite{kailath}.

$$
v(t) = \Phi(t,x)w
$$
where
$$
\Phi(x,t)=I+
\int_x^tQ(\tau_1)d\tau_1+
\int_x^t\int_x^{\tau_1}Q(\tau_1)Q(\tau_2)d\tau_2d\tau_1+\dots
$$

\subsection{The periodic case}\label{periodic}

\begin{theorem}
Assume that $A(t)$ is a 1-periodic  probability matrix that is piecewise continuous in $t$, and that $s(t)$ a 1-periodic function.
Then the equation
\begin{equation}
v^{\cdot} = (A(t)-I) s(t) v
\end{equation}
has a 1-periodic solution.\\
\end{theorem}

This is a consequence of Floquet theory:\\
For all $t$, the adjoint matrix $A(t) ^*$ has eigenvector $(1,1,...1)$ with eigenvalue $1$. Therefore the adjoint equation
\begin{equation}
-y^{\cdot} = (A(t)-I)^*s(t) y
\end{equation}
has $(1,1,...1)$ as a constant and hence periodic solution.
Following \cite{knobloch} p. 94, the space $L$ of periodic solutions of $(1)$ has the same dimension as the space $L^*$ of periodic solutions of (2). Therefore the space $L$ has dimension at least one.

\begin{theorem}\label{unique}
In this situation, the (normalized) periodic solution is unique.
\end{theorem}
The proof is given in Appendix 2.\\

If $v=(v_1,v_2,...,v_n)$ is the solution of the equation, then
the output density of the neuron is given by the density
$$
p(t) v_1(t)s(t)= s_+(t) v_1(t)
$$
that controls the passage from the first  to the n-th level of the neuron.

\subsection{Interspike interval distribution}

Given a vector $w$ representing the distribution of probability over the states at
the initial time $x$,
we determine the transition probability $p$ given
the time-dependent Poisson processes $N_+$ and $N_-$ as input.
By adding an absorbing state in the diagram
\ref{state_diagram} we find the
infinitesimal generator matrix
\begin{equation*}
Q_{0}(t) =
\begin{bmatrix}
0 & {s}_+(t) &  & & & & \\
0 &-s (t) & {s}_+(t)  &   & & &  \\
0 &{s}_-(t)  & -s (t) & {s}_+(t) & & &  \\
\vdots &  &      \ddots  & \ddots  & \ddots& \\
0 &  &      &           {s}_-(t)& -s (t)& {s}_+(t) \\
0 &  &      &                         & {s}_-(t) & -{s}_+(t) \\
\end{bmatrix},
\end{equation*}
where $s (t)=s_+(t)+s_- (t)$.
The time development of the finite state system up to the next firing is the solution
to the differential equation
\begin{equation}
\label{gen}
v'(t) = Q_{0}(t)v(t)
\end{equation}
The solution $v(t)=(v_0(t),v_1(t),...v_n(t))$ of the differential equation with initial condition $v(x) = e_n$ is the probability distribution over the states, given that the neuron fired at time $t_0=x$ (at this time the neuron is at state $n$). The component $v_0(t)$ is the probability that the next firing occurred in the interval $(x,t]$. The transition probability is therefore given by
\begin{equation}\label{trans}
p(x,t)=\frac{d}{dt}v_0(t) = s_+(t)v_1(t)
\end{equation}
The output density $s^*$ (that depends on $x$ and on $t$) of the neuron can then be calculated as the instantaneous firing rate
$$s^*(x,t)= \sum_1^{\infty} p_m(x,t)$$
(cf. formula (10)).\\

\subsection{Refractory period for the integrate and fire neuron}

The model for the integrate and fire neuron can also be considered as a renewal process with
time dependent input. The process depends on the last firing time $x$ and on the excitatory and inhibitory
inputs densities $s_+$ and $s_-$.
As in section 4.4 it is then possible to incorporate a refractory period in the model. 
The refractory function $r(t-x)$, translated to the position
$x$ has to be multiplied
with the input densities. The differential equation $v'(t) = Q(t) v(t)$
is then solved with
$s_+(t)$ replaced by $s_+(t)r(t-x)$, $s_-(t)$ by $s_-(t)r(t-x)$ and consequently $s(t)$ by $s(t)r(t-x)$.
The transition probability is again given by (\ref{trans}), with $v_1(t)$  the component of the
solution vector of the modified differential equation.

\subsection{Constant stimuli}

In general, the explicit solution of the differential equation is difficult to find. This can  already be seen by considering the simple
case, where the system is time-independent, i.e. the excitation and the inhibition are homogeneous
Poisson processes. Then the Peano-Baker series simplifies to matrix exponentiation. However,
even this cannot be calculated explicitly for large matrices.

\subsubsection*{Two state system}
The simplest case is the two-state system which fires at the moment $t=0$.
There,
the inter spike interval density can be calculated (using Mathematica)
directly and we have
$$
f(t) =
\frac{s _+^2}{\sqrt{s_-(4s_++s_-)}}
\left({\rm e}^{t\sqrt{s_-(4s_++s_-)}}-1\right)
{\rm e}^{-\frac{t}{2}\left(2s_++s_-+\sqrt{s_-(4s_++s_-)}\right)}.
$$

Similar calculations reveal that the instantaneous firing rate
with initial state $w=(0,0,\dots,0,1)$ at moment $t=0$
gives
$$
I(t)=\frac{s_+^2}{2s_++s_-}\left(1-{\rm e}^{-(2s_++s_-)t}\right).
$$
For larger matrices the formulas are similar but become enormous. The important phenomenon is
that the instantaneous firing rate always approaches at an exponential rate a unique equilibrium.

\subsubsection*{Many state system with excitation only}

If inhibition is absent,
then the density of the first passage time of the integrate--and--fire neuron is
$$
f(t) = s_+(\exp(Q_{0}t)w)_1  = \frac{ \exp(-s_+t)(s_+t)^{K}}{{K}!t},
$$
with the choice $w = (0,0,\dots, 0, 1)$.
We find thus
the well known formula for perfect integrate-and-fire neurons with excitation only.
\subsubsection*{Equilibrium in the presence of inhibition}
In the equilibrium, the $K$-state system satisfies $Qv = 0$. This leads to
$$
E(t)=\frac{\sum_{k=1}^{K}(K-k+1)s_+^{K-k}s_-^{k-1}}{s_+^{K}}
$$
For $s_+>>s_-$ the excitatory term dominates and the expectation is roughly $K/s_+$;
for $s_+=s_-$ the expectation is
$$
E(t)= \frac{K(1+K)}{2s_+};
$$
for $s_+<<s_-$ the expectation blows up.

\section{Proofs for the theorems}
\section*{Appendix 1}
\begin{theorem}
Let $X$ be a Poisson process with a refractory time $\rho>0$ and constant
density $s(t)=A>0$. Then the instantaneous firing rate is well defined and
$$
q(t)=\frac{A}{1+A\rho}.
$$
\end{theorem}

The Fourier transform of
\begin{equation*}
p(t)=A \chi_{\{ t>\rho\}} e^{-A(t-\rho)}
\end{equation*}
is
\begin{eqnarray*}
\hat{p}(\omega)&=&A \int_{\rho}^{\infty} e^{-i\omega t} e^{-A(t-\rho)}dt =
 A e^{-i\omega \rho}\int_{0}^{\infty}e^{-(A+i\omega)t}dt\\
&=&\frac{A e^{-i\omega \rho}}{A+i\omega}
\end{eqnarray*}
By the convolution theorem

\begin{equation*}
\hat{p}_k(\omega):=\hat{p^{*k}}(\omega)=
A^k\frac{\e ^{-ik\omega\rho}}{(A+i\omega)^k}
\end{equation*}
Since $p_k(x)\geq 0 $ has support on $[k\rho,\infty )$ and $\int p_k=1$, we deduce that \\ $\int_{-\infty}^{\infty}\big|q(t)| (1+|t|)^{-2}dt <\infty, $
where $q$ is the locally integrable function
$$
q(t):=\sum_{k=1}^{\infty}p_k(t),
$$
and its Fourier transform satisfies
$$
\widehat q(\omega):=\sum_{k=1}^{\infty}\hat{p}_k(\omega)\qquad \textrm{convergence in}\; \mathcal{S}'.
$$
We may perform the summation
 in terms of geometric series  for $\omega \not=0$ (with locally uniform convergence outside the origin), and obtain
$$
\widehat q(\omega) = \frac{\e ^{-i\omega \rho}}{1+\frac{i\omega}{A}-\e ^{-i\omega \rho}}, \qquad \omega\not=0.
$$

The function  $\omega\mapsto \frac{\e ^{-i\omega \rho}}{1+\frac{i\omega}{A}-\e ^{-i\omega \rho}}$ is not integrable at the origin and hence the classical definition of
the Fourier transform does not apply. However, it certainly defines a distribution in the principal value sense
and below when dealing with this function as a distribution we assume that it is defined as a principal value sense.
Since the Fourier transform is unique, $\hat{q}$ will then have the form
$$
\hat{q}=\hat{d} + \{ p.v.\}  \frac{\e ^{-i\omega \rho}}{1+\frac{i\omega}{A}-\e ^{-i\omega \rho}},
$$
where $\hat{d}$ is a distribution supported at origin. Our aim is to show that $d$ is a constant.

Consider the behaviour of $\hat{q}-\hat{d}$.
When $\omega$ is close to zero we have

$$
\hat{q}(\omega)-\hat{d}(\omega)
=
\frac{A}{1+A\rho}\frac{1}{i\omega}+O(1).
$$

Since the term $\e ^{-i\omega (\rho)}$ in the numerator corresponds to a shift
we neglect it for the moment. For big values of $\omega$
$$
\frac{1}{1+\frac{i\omega}{A}-\e ^{-i\omega \rho}}
=
\frac{A}{i\omega+O(1)}.
$$
Let us write
$$
\frac{1}{1+\frac{i\omega}{A}-\e ^{-i\omega \rho}}
=
\hat{R}(\omega) + \hat{S}(\omega)
$$
with
$$
\hat{R}(\omega)=\frac{A(1+i\omega)}{i\omega(1+A\rho+i\omega)}.
$$
For small values of $\omega$ we see that
$$
\hat{S}(\omega)
=
\left(\frac{A}{1+A\rho}\right)^2\frac{\rho ^2-2\rho}{2}+O(\omega)
$$
and for large values of $\omega$ we have
$$
\hat{S}(\omega)
=
\frac{A^2(1-\rho-\e^{-i\omega\rho})}{\omega ^2+O(\omega)}.
$$
Hence $\hat{S}$ is integrable and by the Riemann-Lebesgue lemma the inverse
Fourier transform is uniformly continuous and vanishes at infinity. Since
$$
\hat{R}(\omega)
=
\frac{A}{1+A\rho}\frac{1}{i\omega}+\frac{A^2\rho}{1+A\rho}\frac{1}{1+A\rho+i\omega}
$$
we can compute the inverse Fourier transform (recall that we need to interpret the distributions at the origin in the principal value sense)
as
$$
R(t)
=
\frac{A^2\rho}{1+A\rho}\e ^{-(1+A\rho)t}\chi_{(0,\infty)}(t)
+
\frac{A}{2(1+A\rho)}{\rm sgn}(t).
$$
But  we know that $q(t)$ is zero for $t<\rho$. Also, every distribution $\hat{d}$
supported at the origin has an inverse Fourier transform, which is a (complex)
polynomial. Since
$$
q(t) =
d(t) + R(t-\rho) + S(t-\rho)
$$
we conclude that
$d$ must be constant. Hence
$$
d=\frac{A}{2(1+A\rho)}
$$
and
$$
q= \lim_{t\rightarrow -\infty}q(t)=\frac{A}{1+A\rho}
$$
This proves the theorem. For the Fourier analysis results we refer to \cite{S}.

\section*{Appendix 2}
In this appendix we give a full proof of theorem \ref{unique}. The proof is a fixed point
argument. We show that the statement of the theorem is equivalent to finding a unique fixed point
of an integral operator. This operator then is shown to be positive and by examining the dual
space, we conclude that the fixed point is attractive in the sense that any initial value
will converge to the same fixed point geometrically fast.
Here we recall some classical facts on uniformly ergodic Markov chains.
Let $(S,\B)$ be a measurable space and denote by $\M (S)$ the Banach space of all finite measures on $(S,\B)$, equipped with the total variation norm, denoted by $\|\cdot \|_1$. The subset of probability measures is denoted by $\Prm (S).$  Let us call any linear and bounded  operator $T:\M (S)\to \M(S)$ that maps positive measures to positive measures in an isometric way a {\sl transition operator}. Especially, transition operators map probability measures to probability measures. The Dobrushin coefficient of ergodicity is defined by the formula
\begin{equation}\label{eq:dobrushin}
\delta (T):=\sup\frac{\|T\mu_1-T\mu_2\|_1}{\|\mu_1-\mu_2\|_1},
\end{equation}
where the supremum is taken over all distinct $\mu_1,\mu_2\in \Prm (S).$
It is straightforward to check  from the definition that for products of transition operators one has
\begin{equation}\label{eq:dobrushinproduct}
\delta (STU)\leq \delta (T).
\end{equation}
\begin{theorem}\label{lemma:dobrushin} Assume that the transition  operator $T:\M (S)\to \M (S)$ satisfies $\delta (T) <1$. Then there is a unique equilibrium distribution $\pi\in \Prm (S) $ with $T\pi =\pi .$ Moreover, one has the exponential convergence
$$
\|\pi-T^n\mu\|_1\leq 2(\delta(T))^n,\quad {\rm for \; all }\;\; n\geq 1\quad
{\rm and}\quad \mu\in \Prm (S).
$$
\end{theorem}
\begin{proof}
The existence and uniqueness of $\pi$ follows by Banach's fixed point theorem since the assumption $\delta (T)<1$ states that $T$ is a strict contraction on the closed subset $\Prm (S)\subset \M (S).$ Also, we may estimate for $n\geq 1$
$$
\|\pi-T^n\mu\|_1=\|T\pi-T(T^{n-1})\mu\|_1\leq \delta(T) \|\pi-T^{n-1}\mu\|_1,
$$
and the second statement follows by iteration.
\end{proof}
\subsection*{The refractory neuron}
There is a particular class of transition operators that can be expressed as integral operators and for which the above Lemma applies directly.  In particular, this is the case for the operator in Theorem \ref{fixedpoint} in our setup. Namely, assume that  the space $(S,\B)$
admits a reference probability measure $m\in \Prm (S)$ and the transition  operator $T:\M (S)\to \M (S)$
acts via
$$
T\mu (A)=\int_{A\times S}k(x,y)m (dx) \mu(dy)\qquad {\rm for \; any}\;\; \mu\in \Prm (S),
$$
where the integral kernel $k:S\times S\to [0,\infty)$ is measurable and non-negative.
\begin{theorem}\label{lemma:kernel}
Assume that the  kernel $k$ of the transition operator $T$  satisfies
$$
k(x,y)\geq g(x) \qquad  {\rm for \; all}\;\; x,y\in S,
$$
where $g(x)\geq 0.$ Then $$\delta (T)\leq 1-\int_Sg(x)m(dx).$$
Especially, if the function $g$ is not  zero a.e., then $\delta (T)<1$ and the conclusions of Lemma \ref{lemma:dobrushin} apply.
\end{theorem}
\begin{proof}
Assume $a:=\int_S g(x)m(dx)>0$. Then   $U:\M (S)\to \M (S)$ is a transition operator, where   $U$ has the kernel $\widetilde k(x,y):=a^{-1}g(x)$ for $x,y\in S.$ If we
define $\lambda \in \Prm (S)$ by setting $\lambda (A)=a^{-1}\int_Ag(x)m(dx),$
it follows that
\begin{equation}\label{eq:U}
U\mu = \lambda\qquad {\rm for \; any} \;\; \mu\in \Prm (S).
\end{equation}
The assumption on  $k$ implies that the operator $T-aU$ is positive  and by applying it on probability measures we hence see that it satisfies the norm bound
$$
\|(T-aU): \M(S)\to\M(S)\| =1-a.
$$
Given any two probability measures $\mu_1,\mu_2$ on $S$ we  apply \refeq{eq:U} to compute
$$
\|T\mu_1-T\mu_2\|_1= \|(T-aU)(\mu_1-\mu_2)\|_1\leq (1-a)\|\mu_1-\mu_2\|_1,
$$
and this proves the Theorem.
\end{proof}

The first part of the following corollary yields an alternative proof for Theorem \ref{fixedpoint}. For that end, given a Banach space $E$ and $x_0\in E$, $x'_0\in E'$ the one dimensional operator $y\mapsto \langle x_0',y\rangle x_0$ on $E$ is denoted  as usual by $x'_0\otimes x_0$.
\begin{cor}\label{adjoint} There is $a\in(0,1)$ such that  in the situation of Theorem \ref{fixedpoint} it holds that
\begin{equation}\label{co(i)}
\|\widetilde T^k f-c\widetilde q\|_{L^1(0,1)}\leq (1-a)^k\qquad \textrm{where}\quad c:=\int_0^1f
\end{equation}
 for  all $k\geq 1$ and all functions $f\in L^1(0,1).$

In addition, the fixed points of the dual operator $\widetilde T^*$ consist of constant functions, and we have
\begin{equation}\label{co(ii)}
\|( (\widetilde T^*)^k g- c\|_{L^\infty(0,1)}\leq (1-a)^k\qquad \textrm{where}\quad c:=\int_0^1f\widetilde q
\end{equation}
 for  all $k\geq 1$ and all functions $g\in L^\infty (0,1).$
 \end{cor}
 \begin{proof}  Theorem \ref{lemma:kernel} shows that $\| \widetilde T f-\widetilde q\|_1\leq (1-a)$ for all $f\geq 1$ with $\int_0^1f=1.$ By
 considering separately the positive and negative parts of a general $f\in L^1(0,1)$ we deduce that $\|\widetilde T-1\otimes \widetilde q\|_{L^1\to L^1}\leq (1-a).$ By definition $\widetilde T\widetilde q=\widetilde q$ and $\widetilde T^*1=1$ so that $(1\otimes \widetilde q)\widetilde T=\widetilde  T(1\otimes \widetilde q)=1\otimes \widetilde q$. This yields that $(\widetilde T-1\otimes \widetilde q)^k=\widetilde T^k-1\otimes \widetilde q.$ We thus obtain the first statement by noting that
 $$
 \| \widetilde T^k-1\otimes \widetilde q \|_{L^1\to L^1}= \| (\widetilde T-1\otimes \widetilde q)^k\|_{L^1\to L^1}\leq (1-a)^k.
 $$
By taking adjoints in this inequality we see also that
 $$
 \| (\widetilde T^*)^k-1\otimes \widetilde q \|_{L^\infty \to L^\infty}\leq (1-a)^k,
 $$
 which clearly implies the second part of the Corollary.
 \end{proof}

\subsection*{The integrate-and-fire neuron}
We next apply the above notions to find unique periodic solutions to the periodic differential equations in our model and prove exponential convergence of any solution to the periodic one.
Consider the linear and time-dependent differential equation in $\R^n$
\begin{equation}\label{eq:system}
x'(t)=Q(t)x(t),
\end{equation}
where $x:[0,\infty)\to\R^n$ and $Q:[0,\infty )\to \R^{n\times n}$ is a continuous\footnote{One can easily weaken the assumption of continuity} matrix valued function that is 1-periodic.
We assume that for each $t$ the matrix $Q$ is {\sl probability generating}, i.e. its column sums are zero, diagonal elements are non-positive and other elements non-negative.
\begin{theorem}\label{lemma:floquet} Assume that  for some $0\leq t_0<t_1\leq 1$
and for some constant $c>0$ the 1-periodic probability generating matrix $Q(t)=(q_{i,j})$ in {\rm \refeq{eq:system}} satisfies the condition
$$
q_{j,j-1}(t)\geq c\qquad {\rm for}\;\;  t\in (t_0,t_1), \quad j=1,2,\ldots ,n,
$$
where we interpret $q_{1,0}=:q_{1,n}.$
Then the system \refeq{eq:system} admits a unique 1-periodic  solution (that is a probability distribution)  and all  solutions  (that are probability distributions) converge with exponential rate towards this solution as $t\to\infty$.
\end{theorem}
\begin{proof}
Denote by $\Phi(t,s)$ the transition (operator) matrix of the flow map, so that  $x(t)=\Phi(t,s)x(s)$
for all $0\leq s\leq t$ regardless of the initial values.
We then know {\em a priori} that the coefficients of $\Phi$ cannot be negative, since solutions with initial condition inside the positive cone stay inside the positive cone.
For the proof of the lemma it is enough to show that
\begin{equation}\label{eq:enough}
a:=\delta\big(\Phi(t_1,t_0)
\big) <1.
\end{equation}
Namely, then  \refeq{eq:dobrushinproduct} and the equality $\Phi(1,0)
= \Phi(1,t_1,) \Phi(t_1,t_0) \Phi(t_0,0)$ shows that
$$
\delta \big(\Phi(1,0)\big)\leq a<1,
$$
 whence Theorem \ref{lemma:dobrushin} verifies that there is a unique probability vector (probability distribution on $\{ 1,2,\ldots n\}$) that solves $\Phi (1,0)\pi =\pi$. This implies that there is a unique 1-periodic solution to the flow. Finally,  for any integer $m\geq 1$ the  periodicity of $A$ yields that $\Phi (m,0)= (\Phi(1,0))^m$, whence the second statement of Theorem \ref{lemma:dobrushin}
shows that $\| x(m)-\pi\|_1\to 0$ with exponential rate and this clearly yields the stated exponential convergence of the flow to the periodic solution.
By Theorem \ref{lemma:kernel} it is enough to check that the transition matrix $\Phi (t_1,t_0)$ satisfies $\Phi (t_1,t_0)_{j,k}>0$ for all $1\leq j,k\leq n.$ Equivalently, we need to show that
\begin{equation}\label{eq:enough2}
x_j(t_1)>0 \quad {\rm for\; all}\;\; j\in\{ 1,2,\ldots ,n\},
\end{equation}
where $x$ solves the flow $x'=Q(t)x$ on $t\in[t_0,t_1]$ with initial condition
$x_j(t_0)=1$ for $j=k$, and $x_j(t_0)=0$ for $j\not=k$. By the cyclic symmetry of the situation we may as well assume that $k=1$.
 Then the first equation and the  assumed properties of $Q(t)$
yield the differential inequality (the constants $c,C>0$ below are fixed for all $t\in[t_0,t_1]$)
$$
x_1'\geq  -Cx_1 \quad {\rm with}\;\;  x_1(t_0)=1
$$
which  shows that
\begin{equation}\label{eq:x1}
x_1\geq b_1\qquad {\rm for}\;\;  t_0\leq t\leq t_1,
\end{equation}
where $b_1>0$ depends only on $C$ and $t_1-t_0$. Then the last equation implies
$$
x_n'\geq cx_1-Cx_n\geq cb_1-Cx_n \qquad {\rm with }\;\;  x_n(t_0)=0
$$
and  standard integration gives
\begin{equation}\label{eq:xn}
x_n(t)\geq b_n(t-t_0)\qquad {\rm for}\;\;  t_0\leq t\leq t_1.\nonumber
\end{equation}
Iterating, the  differential inequalities $x'_k\geq  cx_{k-1}-Cx_k$  finally produce
the bound
 \begin{equation}\label{eq:xk}
x_k(t)\geq b_k(t-t_0)^{n+1-k}\qquad {\rm for}\;\;  t_0\leq t\leq t_1\quad {\rm and}\;\; 2\leq k\leq n.
\end{equation}
The statement  \refeq{eq:enough2} follows from   \refeq{eq:x1} and \refeq{eq:xk}, where one notes that the obtained bound depends only on $n,c,C$ and
$t_1-t_0$.
\end{proof}
In the applications, the infinitesimal generator of the process is
given by the matrix\begin{equation*}
Q(t) =
\begin{bmatrix}
-s(t) & {s}^+(t)  &   & & &  \\
{s}^-(t)  & -s(t) & {s}^+(t) & & &  \\
  &      \ddots  & \ddots  & \ddots& \\
  &      &           {s}^-(t)& -s(t)& {s}^+(t) \\
{s}^+(t) &      &            & {s}^-(t) & -{s}^+(t) \\
\end{bmatrix},
\end{equation*}
with periodic and piecewise continuous  functions $s^+$, $s^-$ and $s= s^+ + s^-$. The transition mapping $\Phi = \Phi(0,1)$
then satisfies the required condition
$$\delta(\Phi) < 1$$
Indeed, the conditions of the previous theorem are clearly satisfied in our setup.

\subsection*{Expectation values}
For a periodic density $s(t)$, the expected time span between two consecutive firings, given that the neuron fired at $T_0=x$ is $$
E(x)= \int (t-x) p(x,t) dt
$$
The expectation $E(x)$ is a periodic function
\begin{eqnarray*}
E(x+1)&=& \int (t-x-1) p(x+1,t) dt\\
&=& \int (t-x) p(x,t-1) dt\\
&=& \int (t-x) p(x,t) dt = E(x)
\end{eqnarray*}
The expected $k$:th firing delay is calculated recursively from
\begin{eqnarray*}
E_1(x)&=& E(x)\\
E_k(x)&=&  \int (t-x) p_k (x,t) dt\\
&=&  \int (t-x) \int p(x,z) p_{k-1}(z,t) dz dt\\
&=&\int \int p(x,z)(t-z) p_{k-1}(z,t) dz dt +\int \int (z-x) p(x,z) p_{k-1}(z,t) dz dt\\
&=&\int E_{k-1}(z)p(x,z) dz + E_1(x)\\
&=& T^*E_{k-1}(x) + E_1(x)
\end{eqnarray*}
As a consequence
\begin{eqnarray*}
E_k(x) -E_{k-1}(x) &=& T^*(E_{k-1}(x) - E_{k-2}(x)) \\
  &=& T^{*k-2} (E_2(x) - E_1(x))\\
  &=& T^{*k-1} E(x)
\end{eqnarray*}
As before, we shall denote by $\widetilde g$ the restriction of a 1-periodic function on $\R$ to the interval $(0,1).$
By \eqref{co(ii)} the following limit exists in $L^\infty(0,1)$
\begin{equation*}
\widetilde m(x)= \lim_{k\rightarrow \infty}(\widetilde T^*)^{k} E(x) =\lim_{k\rightarrow \infty}(\widetilde E_k(x)-\widetilde E_{k-1}(x)),
\end{equation*}
and the rate of convergence is exponential.
Then $\widetilde m(x)$ is a fixed point of $\widetilde T^*$, and Corollary \ref{adjoint} verifies that it is a constant function.
Function  $m(x)$ is  periodic
hence it is a constant.

Note that
\begin{eqnarray*}
\lim_{k\rightarrow \infty}\frac{E_k(x)}{k}&=& \lim_{k\rightarrow \infty}\frac{1}{k}(T^{*k-1}E_1(x) + ... + E_1(x))\\
&=& m
\end{eqnarray*}
Both the differences $ E_k(x)-E_{k-1}(x)$ and the mean $ \frac{E_k(x)}{k}$ thus converge uniformly to the constant $m$.
If $q$ is a 1-periodic fixed point of $T$ and hence of $\tilde{T}$, then
\begin{eqnarray*}
(q,m) &=& m \int_0^1 q(x) dx\\
&=& \lim_{k\rightarrow \infty} (q,\tilde{T}^{*k} E) \\
&=& \lim_{k\rightarrow \infty}(\tilde{T}^{*k} q, E)\\
&=& (q,E)
\end{eqnarray*}
The constant is then given as
$$
m=\frac{\int_0^1 q(x) E(x) dx}{\int_0^1 q(x)dx}
$$
We have thus shown:
\begin{theorem}
Assume that $E\in L^{\infty}$.
Both the differences $E_K -E_{k-1}$ and the mean $\frac{E_k}{k}$ converge uniformly and exponentially
towards the constant m.
\end{theorem}

Remark.\\
Theorem \ref{constant} shows that for the refractory neuron

\begin{equation*}
\int_0^1 q(x) E(x) dx = 1
\end{equation*}

In the absence of a refractory period this is easily proved:
\begin{eqnarray*}
E(x)&=& \int_x^{\infty} (t-x) p(x,t) dt\\
&=&\int_x^{\infty} (t-x) (-\frac{d}{dt} e^{-\int_x^t s(u) du}) dt\\
&=&\int_x^{\infty} e^{-\int_x^t s(u) du} dt
\end{eqnarray*}
\begin{eqnarray*}
\frac{d}{dx}E(x)&=& -e^{-\int_x^x s(u)du} +  \int_x^{\infty} s(x) e^{-\int_x^t s(u) du} dt\\
&=& -1 + s(x) E(x)
\end{eqnarray*}
Since $E(x)$ is periodic
\begin{equation*}
\int_0^1 s(x) E(x) dx = \int_0^1 dx + \int_0^1 \frac{d}{dx}E(x) dx = 1
\end{equation*}
It has already been shown that $q(t) =s(t)$ in this special situation.

\begin{thebibliography}{X}

\bibitem{cox}D.R. Cox, Renewal Theory. Methuen, 1962.

\bibitem{daley}D. J. Daley, D. Vere-Jones. An introduction to the Theory of Point Processes. Volume II: General Theory and Structure. Second edition, Springer, 2008.


\bibitem{deger}M. Deger, M. Helias, S. Cardanobile, F. M. Atay, and S. Rotter, Nonequilibrium dynamics of stochastic point processes with refractoriness. Phys. Rev. E 82 (2010).

\bibitem{ermentrout2011} G. B. Ermentrout, D. H. Terman, Mathematical Foundations of Neuroscience.
Springer, 2011.

  \bibitem{F2} W. Feller: \emph{ An introduction to probability theory and its applications.}  Vol. II. Second edition John Wiley \& Sons, Inc., 1971.


\bibitem{fortmann} T. E. Fortmann, and K. L. Hitz, An introduction to linear control systems. Marcel Dekker, New York, 1977.

\bibitem{gaumond82} R. P. Gaumond, C. E. Molnar, and D. O. Kim, Stimulus and recovery dependence of cat cochlear nerve fiber spike discharge probability. J. Neurophys. 48
(1982), 856--873.

\bibitem{gerstnerkistler}W. Gerstner, and W. Kistler, Spiking Neuron Models. Cambridge, 2002.

\bibitem{gihman69} I. I. Gihman, A. V. Skorohod, Introduction to the teory of random processes, W. B. Saunders Company, Philadelphia, 1969.


\bibitem{kailath}T. Kailath, Linear systems, pp. 594--631. Englewood Cliffs, Prentice-Hall, 1980.

\bibitem{knobloch} H.W. Knobloch und F. Kappel, "Gew\"{o}hnliche Differentialgleichungen" Teubner, Stuttgart, 1974.
\bibitem{lutkenhoner80}B. L\"utkenh\"oner, Threshold and beyond: modeling the intensity dependence of auditory responses  JARO 9 (2008), 102121 .


\bibitem{miller92} M. I. Miller, and K. E. Mark, A statistical study of cochlear nerve discharge patterns in response to complex speech stimuli.J. Acoust. Soc. Am. 92 (1), (1992).

\bibitem{rindos} A. Rindos, S. Woolet, I. Viniotis, and K. Trivedi. Exact methods for the transient analysis of nonhomogeneous continuous time Markov chains. Numerical Solution of Markov Chains, 121--134, 1995.

\bibitem{stein65} R. B. Stein, A theoretical analysis of neuronal variability. Biophysical Journal, 5 (1965), 173--194.

\bibitem{S} R. Strichartz, A Guide to Distribution Theory and Fourier
  Transforms, CRC Press, Boca Raton, 1994.

\bibitem{thorisson84} H. Thorisson, Periodic regeneration. Stochastic Processes and their Applications 20 (1984) 85-104.

\bibitem{tuckwell89} H. C. Tuckwell, Stochastic Processes in the Neurosciences. Siam, 1989.

\bibitem{inferiorcolliculus}J. A. Winer, C. E. Schreiner, The Inferior Colliculus. Springer, 2005.
\end{thebibliography}
\end{document}